\documentclass[aps,pra,nofootinbib,onecolumn,longbibliography,superscriptaddress,tightenlines,notitlepage,11pt]{revtex4-2}
\pdfoutput=1 

\usepackage{graphicx}
\usepackage{amsmath}
\usepackage{amssymb}
\usepackage{amsfonts}
\usepackage{amsthm}
\usepackage{mathtools}
\usepackage{bm}
\usepackage{hyperref}
\usepackage{braket}
\usepackage[normalem]{ulem}
\usepackage{wrapfig}
\usepackage{tikz}

\hypersetup{colorlinks=true,urlcolor=[rgb]{0,0,0.5},citecolor=[rgb]{0.5,0,0},linkcolor=[rgb]{0,0,0.4}}

\newtheorem{theorem}{Theorem}
\newtheorem{problem}{Problem}
\newtheorem{corollary}{Corollary}
\newtheorem{lemma}{Lemma}
\newtheorem{proposition}{Proposition}
\newtheorem{definition}{Definition}

\def\autorefapp#1{\hyperref[#1]{Appendix~\ref{#1}}}

\newcommand\emails[1]{\begingroup
\renewcommand\thefootnote{}\footnote{#1}
\addtocounter{footnote}{-1}\endgroup}
    
\def\ep{\varepsilon}

\def\vphi{\varphi}

\def\tr{{\rm tr}}

\def\vev#1{\langle{#1}\rangle}

\def\and{\quad {\rm and} \quad}

\def\ni{\noindent}
\def\nn{\nonumber\\}

\def\ie{{\rm i.e.\ }}
\def\eg{{\rm e.g.\ }}

\def\ketbra#1{ |{#1}\rangle\!\langle{#1}| }
\def\O{\mathrm{O}}

\DeclareMathOperator*{\Ex}{\mathbb{E}}
\DeclareMathOperator*{\Var}{{\rm Var}}

\begin{document}

\title{Ideal random quantum circuits pass the LXEB test}
\author{Nicholas Hunter-Jones}
\affiliation{Department of Physics, University of Texas at Austin, Austin, TX}
\affiliation{Department of Computer Science, University of Texas at Austin, Austin, TX}
\author{Jonas Haferkamp}
\affiliation{Department of Mathematics, Saarland University, Saarbruecken, Saarland}
\affiliation{School of Engineering and Applied Sciences, Harvard University, Cambridge, MA}

\emails{\hspace*{-2mm} \href{mailto:nickrhj@utexas.edu}{\tt nickrhj@utexas.edu}, \href{mailto:haferkamp@math.uni-sb.de}{\tt haferkamp@math.uni-sb.de}}

\begin{abstract}
We show that noiseless random quantum circuits pass the linear cross-entropy benchmark (LXEB) test with high probability. If the circuits are linear depth, and thus form unitary 4-designs, the LXEB test is passed with probability $1-O(1/\sqrt{k})$, where $k$ is the number of independently drawn samples from the output distribution of the random circuit. If the circuits are of depth $\tilde O(n^2)$, and thus form unitary $n$-designs, the LXEB test is passed with probability $1-O(e^{-k \log(n)/n})$. In proving our results, we show strong concentration of the random circuit collision probability at linear depth and establish that the tails of the distribution of random circuit output probabilities start to resemble Porter-Thomas at near-quadratic depths. Our analysis employs higher moments and high-degree approximate designs.
\end{abstract}

\maketitle

\section*{Introduction}
In recent years, control of programmable quantum systems has rapidly advanced, bringing us close to performing computations on quantum devices that cannot be done classically. To definitively demonstrate that existing noisy quantum processors can outperform classical computers, various sampling schemes have been introduced and implemented \cite{aaronson2010computational,bremner2016average,morimae2017hardness,boixo2018characterizing}. This includes experimental demonstrations with 50 to 80 noisy qubits \cite{arute2019quantum,wu2021strong,zhu2021quantum,morvan2023phase,decross2025comp,gao2025establishing} and with more than 100 photons \cite{zhong2020quantum,madsen2022quantum}. The sampling problems implemented on these quantum devices are not believed to be classically solvable in polynomial time, based on complexity-theoretic considerations. In fact, experiments based on random circuit sampling \cite{arute2019quantum} are now in a prolonged battle with classical simulations, shifting the benchmark for quantum supremacy \cite{pan2022simulation,pan2022solving}. 
Should quantum processors eventually win the race for good, this would constitute an experimental violation of the extended Church-Turing thesis, which posits that every model of computation implementable in nature should be efficiently simulable on a classical computer.

Beyond this proof-of-principle demonstration, the data acquired in sampling random quantum circuits is not devoid of practical applications.
Indeed, Aaronson and Hung~\cite{aaronson2023certified} proposed to use the generated samples to produce certified bits of randomness; also, see Ref.~\cite{bassirian2021certified}.
Certified randomness is directly relevant for practical problems such as proof-of-stake cryptocurrencies. 
More precisely, it was proven in Ref.~\cite{aaronson2023certified} that the outputs of random quantum circuits contain $\Omega(n)$ min-entropy if the Linear Cross Entropy Benchmark (LXEB) is passed.
In fact, a version of this proposal was experimentally tested \cite{liu2025certified}.
A different protocol for certified randomness amplification was also recently demonstrated \cite{liu2025certified2}.

The LXEB is now a standard verification technique for random quantum circuit sampling, as in the Google experiment \cite{arute2019quantum}. See the review article \cite{hangleiter2023comp} for an in-depth discussion of the LXEB and random circuit sampling. The LXEB plays a few roles in modern random circuit sampling experiments, where achieving a nontrivial LXEB value is potentially computationally intractable \cite{boixo2018characterizing,aaronson2020classical} and where the LXEB fidelity might serve as a proxy for the fidelity \cite{arute2019quantum}.
In both cases, the limitations and advantages of this quantity have been studied extensively \cite{gao2021limitations,barak2020spoofing,aharonov2023poly,manole2025howmuch}.

In the LXEB, an implementation of random quantum circuit sampling is accepted if it satisfies the following test, as defined in Ref.~\cite{aaronson2023certified}.
\begin{problem}[LXEB test]\label{prob:lxeb}
    Pick a positive integer $k$ and $b\in (1,2)$. Given a unitary $U$ drawn from a probability distribution on $SU(2^n)$, generate samples $x_1,\ldots,x_k\in\{0,1\}^n$ from the output distribution $p_U$ such that 
    \begin{equation}
        \frac{1}{k}\sum_{i=1}^k p_U(x_i)\geq \frac{b}{2^n}\,,
    \end{equation}
    where $p_U(x) = |\!\braket{x|U|0}\!|^2$.
\end{problem}

It can be easily shown that this test is satisfied by the output distributions of Haar random unitaries \cite{boixo2018characterizing}, also see \cite{kliesch2021theory,hangleiter2023comp}.
More precisely, the output distributions of Haar random unitaries are approximately Porter-Thomas distributed.
However, an implementation of the Haar measure requires exponentially-deep quantum circuits and is therefore unrealistic.
Fortunately, numerical evidence suggests that the output distributions of linear-depth random quantum circuits look approximately Porter-Thomas~\cite{boixo2018characterizing}.
As yet, it is unknown in any rigorous sense if the output distributions of polynomial (or linear) depth random circuits are actually close to the Porter-Thomas distribution (for instance, in total variation distance).
Furthermore, even a proof that random quantum circuits pass the linear cross-entropy benchmark has been missing.
In this paper, we close this gap in the literature and prove that noiseless random quantum circuits pass the LXEB test.
We emphasize that our results do not imply closeness to Porter-Thomas, which we leave as an intriguing open problem.

\section*{Results}

We prove the following two theorems, the first applicable once the random quantum circuits reach a depth linear in the system size:

\begin{theorem}[Linear depth random quantum circuits]\label{thm:lindepth} An $n$-qubit random quantum circuit of depth $144n$ passes the LXEB test with probability $\geq 1 - O(1/\sqrt{k})-O(1/2^n)$, where $k$ is the number of samples generated.
\end{theorem}

\ni The second applies when the circuit depth is polynomial (almost quadratic) in the system size:
\begin{theorem}[Polynomial depth random quantum circuits]\label{thm:polydepth} A random quantum circuit of depth $O(n^{2}\,\mathrm{polylog}(n))$ passes the LXEB test with probability $\geq 1 - O(e^{-k \log(n)/n}) - O(1/2^n)$.
\end{theorem}

The two theorems are more precisely stated and proved as \autoref{thm:lindepthfrml} and \autoref{thm:polydepthfrml} in the main text of this work. Our proofs make use of the convergence of local random quantum circuits to approximate unitary $t$-designs. The proof of \autoref{thm:lindepth} employs approximate unitary 4-designs, which are achieved when the circuits are of linear depth \cite{BHH12,haferkamp2022random,chen2025incompressibility,HHJ20}. Moreover the constants involved are known to be small. \autoref{thm:polydepth} makes use of unitary $n$-designs, which require random quantum circuits of polynomial depth. Specifically, local random quantum circuits form approximate $n$-designs in a depth $O(n^{2}\,\mathrm{polylog}(n))$ \cite{chen2025incompressibility}, from which the theorem follows.

More structured models such as the coarse-grained circuits in Refs.~\cite{schuster2024random,laracuente2024designs} converge to relative-error designs in log depth. Consequently, for these ensembles, we can save a factor of $n$ up to log factors in the above theorems. We make these statements precise in~\autorefapp{sec:otherensembles}. While it is known that a number of random circuit quantities reach their Haar value at log depth \cite{BF13,dalzell2022anticon,dalzell2024noise}, it is currently unknown if the standard brickwork random quantum circuits we consider here are relative-error $t$-designs at log depths. However, recent evidence suggests a similar convergence speed in the unstructured case: Ref.~\cite{heinrich2025anti} shows that brickwork random quantum circuits become relative-error state 2-designs in log depth and Ref.~\cite{laracuente2025quantum} shows $\mathrm{polylog}$-convergence to additive-error $t$-designs. We therefore suspect that the depth in~\autoref{thm:lindepth} and~\autoref{thm:polydepth} is not optimal. 

We also prove a slightly stronger result for linear depth circuits which uses higher moments. In \autoref{prop:8design}, we show that once our random circuits form unitary $8$-designs, they pass the LXEB test with probability greater than $1-O(1/k)$. While convergence to $8$-designs occurs in linear depth, the precise spectral gaps of the moment operator that result in small constants for the fourth moment are not known for eighth moments. Thus, we resort to the more general results in \cite{BHH12,haferkamp2022random,chen2025incompressibility} along with their large constants.

We show that the collision probability concentrates exponentially well, even for circuits of linear depth, using 4-designs. This exponential concentration gets even stronger using higher moments.
It is noteworthy that the concentration result we obtain from the $4$th moment calculation is stronger than anything we could obtain from typical concentration of measure results such as Levy's lemma. Thus, it follows that with high probability over circuits the expectation of the output probabilities is $\Ex_{x\sim p_U}[p_U(x)]\approx 2/2^n$, up to exponentially small error. Next we show that if we generate $k$ independent samples $x_1,\ldots,x_k$ from the output distribution $p_U$, the sample mean of the probabilities is concentrated around its expected value by computing $\Ex_U[\Var_{x}(p_U(x))]$, the variance over samples averaged over circuits. Using standard concentration inequalities, this is enough to establish that once the random circuits form approximate 4-designs at linear depth, the empirical mean is $\frac{1}{k}\sum_i p_U(x_i)\approx 2/2^n$ with probability greater than $1-O(1/\sqrt{k})$. 
The concentration bounds on the collision probability are not sufficient to prove the scaling $1-e^{-\Omega(k)}$ as expected from the Porter-Thomas distribution. 
Instead, we show that the output probabilities are extremely flat in that the maximal output probability is almost always $O(n/2^{n})$. By approximating the $l_{\infty}$-norm by high $l_{p}$-norms, we can obtain exactly this behavior for an approximate $n$-design, and thus random circuits of polynomial depth.
\begin{theorem}\label{thm:maxpUinf}
    When the random quantum circuits are of depth $O(n^{2}\,\mathrm{polylog}(n))$, with probability greater than $1-O(2^{-n})$, the largest output probability $\max_x p_U(x)$ is $O(n/2^n)$.
\end{theorem}
\ni This theorem is precisely stated and proven as \autoref{thm:maxpU}.

Given the results of Ref.~\cite{aaronson2023certified}, this immediately implies that the min-entropy of the output samples of polynomial depth random quantum circuits is $n-O(\log(n))$.
This behavior for the max output probability at polynomial depths is optimal, and will not change for deeper circuits, as this is precisely the $\max_x p_U(x)$ for unitaries sampled from the Haar measure \cite{kliesch2021theory,aaronson2023certified}. 

Using upper bounds on a random variable, we can turn to familiar concentration inequalities to bound the deviation of the sample mean of output probabilities away from its expectation. Along with the exponential concentration of the collision probability around the Haar value and standard concentration tricks, we can then use Bennett's inequality to prove that we pass the LXEB test at polynomial depths with probability greater than $1-O(e^{-k\log(n)/n})-O(1/2^n)$. Using Bennett's inequality gives a stronger bound on the probability than that implied by either Hoeffding or Bernstein's inequality. 

As we are essentially using bounds on the max output probability to crudely bound arbitrary moments of $p_U(x)$, it's possible that the $n$-dependence in the acceptance probability could be improved to establish stronger convergence by computing higher moments. This is the origin of the gap between \autoref{thm:lindepth} and \autoref{thm:polydepth}, where the number of samples in the latter must overcome a system-size dependent factor. We conjecture that the true scaling of the probability that poly-depth random circuits pass the LXEB is $1-O(e^{-k})-O(1/2^n)$, but such a scaling appears to be more difficult to prove with the standard unitary design techniques.

\section*{Outlook}
In this paper, we provide the first rigorous proof that ideal random quantum circuits satisfy the LXEB test, closing a notable gap in the literature. 
In particular, in combination with Ref.~\cite{aaronson2023certified}, we prove that ideal random circuit sampling with polynomial-depth random quantum circuits can be used to generate certified randomness. 

A key open question is if any of these results can be established in the presence of noise.
Indeed, consider the quantity $\mathbb{E}_{\mathcal{E}}\mathbb{E}_{x\sim p_{\mathcal{E}}} p_{\mathcal{E}}(x)$, where $\mathcal{E}$ is a random channel defined by choosing a random quantum circuit interleaved with local depolarizing noise. 
Multiple papers report a phase transition in the noise parameter~\cite{dalzell2024noise,heinrich2023randomized,ware2023sharp,morvan2023phase}.
In particular, for small noise rates the expected collision probability remains somewhat quantum.
Is the same true for the concentration results we found in this work? 
Answering this question might require one to use mappings to statistical-mechanics models for $t>2$.

While it seems likely that our bounds can be further improved, the use of higher moments is essential.
We know that the $3$-design property is not sufficient to prove the LXEB with high probability. Indeed, the multiqubit Clifford group is a unitary $3$-design \cite{zhu2017multiqubit,webb2015clifford} but the output distributions of the Clifford group are uniform with probability $\approx 0.42$ (see e.g. Appendix D in Ref.~\cite{nietner2023average}).
However, it would be interesting to see if our results for the $8$-design property or the $n$-design property can be further derandomized.

\section*{Proofs of main theorems}\label{app:proofs}
Here we present proofs of the theorems in the main text, mostly relying on approximate unitary designs to bound the quantities of interest. First, we give some preliminary definitions relevant to our discussion.

\subsection{Preliminaries}
We consider $n$ qubit systems such that the total dimension is $d=2^n$. In the following we denote the Haar measure on $SU(d)$ by $\mu_H$. First, we define random quantum circuits as follows:
\begin{definition}[Random quantum circuits]\label{def:rqcs}
    We assume for simplicity that $n$ is even.
    Apply a unitary $U_{1,2}\otimes U_{3,4}\otimes...$ and then a unitary $U_{2,3}\otimes U_{4,5}\otimes...$, where all $U_{i,i+1}$ are drawn from $\mu_H$ on $SU(4)$.
    We refer to the number of such layers as the \textit{depth} of the random quantum circuit.
\end{definition}
Such random quantum circuits, often referred to as `brickwork random quantum circuits' \cite{NRVH16,NVH17,vonKey17,RQCstatmech}, are convenient for precisely estimating constants in our analysis, but the same results straightforwardly hold for any ensemble of random circuits which form approximate unitary $t$-designs.

\begin{definition}[Approximate unitary designs]\label{def:approxdesigns}
A probability distribution $\nu$ on $SU(d)$ is an $\varepsilon$-approximate unitary $t$-design if the $t$-fold channel of $\nu$ satisfies
\begin{equation}
	\left\|\Phi^{(t)}_\nu - \Phi^{(t)}_{\mu_H}\right\|_{\diamond}\leq \varepsilon\,,
\end{equation}
where the $t$-fold channel of a distribution $\nu$ is $\Phi^{(t)}_\nu(\cdot) = \int d\nu(U) U^{\otimes t}(\cdot) U^\dagger{}^{\otimes t}$.
\end{definition}
We will use the following theorem from Ref.~\cite{HHJ20} establishing that random quantum circuits, as defined above, give approximate $4$-designs.
\begin{theorem}\label{thm:4designs}
Random quantum circuits generate an $\varepsilon$-approximate unitary $4$-design in depth $16(4n+\log_2(1/\varepsilon))$.
\end{theorem}
This, in turn, follows from the local random circuit $4$th moment gap and an application of the detectability lemma \cite{HHJ20,BHH12}. Note our definition of approximate design above is slightly weaker than the cited work---a constant, as opposed to exponentially small, additive error---which will be convenient for our purposes. For this reason, as well as the definition of depth in \autoref{def:rqcs}, the constant in \autoref{thm:4designs} differs slightly from Ref.~\cite{HHJ20}.

For the bounds based on convergence of higher moments, specifically approximate $n$-designs, we will use the following result bounding the asymptotic behavior in $t$.
\begin{theorem}[Ref.~\cite{chen2025incompressibility}]\label{thm:tdesigns}
    Random quantum circuits form $\varepsilon$-approximate unitary $t$-designs in a circuit depth $O\big(\log^7(t)(2nt+\log_2(1/\varepsilon))\big)$.
\end{theorem}

One fact we will use repeatedly, is the expression for the $t$-th moments of the output probabilities of Haar random unitaries \cite{nietner2023average}.
\begin{lemma}\label{lem:pUmoms}
    For any positive integer $t$ and for $U\sim \mu_H$, the moments of output probabilities $p_U(x)=|\vev{x|U|0}|^2$ are given as
    \begin{equation}
    \Ex_{U\sim \mu_H}\bigg[\prod_{i=1}^{\ell} |p_U(x_i)|^{\lambda_i}\bigg] = \frac{\prod_{i=1}^{\ell} (\lambda_i)!}{\prod_{i=0}^{t-1}(d+i)}\,.
    \end{equation}
    where the $x_i$'s are all distinct and where $\lambda\vdash t$ is an integer partition of $t$, i.e.\ $\lambda = (\lambda_1,\ldots,\lambda_\ell)$.
\end{lemma}
In particular, we have the following special cases, which we will most often use:
\begin{equation}
    \Ex_{U\sim \mu_H}\big[|p_U(x)|^t\big] = \frac{t!}{\prod_{i=0}^{t-1}(d+i)}\,,\qquad
    \Ex_{U\sim \mu_H}\big[|p_U(x)|^t \,|p_U(y)|^t\big] = \frac{(t!)^2}{\prod_{i=0}^{2t-1}(d+i)}\,,
\end{equation}
where in the second expression $x\neq y$. In general, when we write products output probabilities without a sum, the output bit strings are implicitly taken to be different.

These same expressions for the $t$-th moments of the output probabilities hold, up to an additive error, for approximate unitary $t$-designs.
\begin{lemma}\label{lem:approxp}
    Let $\nu$ be an $\ep$-approximate unitary $t$-design. We then have that the $t$-th moments of output probabilites of a unitary drawn from $\nu$ are upper and lower bounded by the Haar values as
    \begin{equation}
        \frac{t!}{\prod_{i=0}^{t-1}(d+i)} - \ep\leq \Ex_{U\sim \nu}\big[|p_U(x)|^t\big]\leq \frac{t!}{\prod_{i=0}^{t-1}(d+i)} + \ep\,.
    \end{equation}
\end{lemma}

\begin{proof}
Recall that the Haar value of the moments of the output probabilities is
\begin{equation}
    \Ex_{U\sim \mu_H}[p_U(x)^t] = \frac{t!}{\prod_{i=0}^{t-1}(d+i)}\,.
\end{equation}
We can upper and lower bound the deviation from the Haar value as
\begin{align}
\begin{split}
    \Big|\Ex_{U\sim\nu}[p_U(x)^t] - \Ex_{U\sim \mu_H}[p_U(x)^t]\Big| &= \tr\Big( \Big|\ketbra{x}^{\otimes t} \Big( \Phi^{(t)}_\nu -\Phi^{(t)}_{\mu_H}\Big)(\ketbra{0}^{\otimes t})\Big|\Big)\\
    &\leq \Big\| \Big(\Phi^{(t)}_\nu -\Phi^{(t)}_{\mu_H}\Big)(\ketbra{0}^{\otimes t})\Big\|_1\\
    &\leq \big\| \Phi^{(t)}_\nu -\Phi^{(t)}_{\mu_H}\big\|_\diamond \leq \ep\,,
    \end{split}
\end{align}
from which the claim follows.
\end{proof}
Identical expressions hold for any $t$-th moment of the output probabilities. We note that we can also use relative-error designs to bound moments of the output probability.

\subsection{Linear depth random quantum circuits pass the LXEB test}
In this section, we prove that random quantum circuits of depth linear in the number of qubits $n$, pass the linear cross-entropy test with probability greater than $1-O(1/\sqrt{k})-O(1/d)$. We state a formal version of \autoref{thm:lindepth} as follows:
\begin{theorem}[Restatement of \autoref{thm:lindepth}]\label{thm:lindepthfrml}
    An $n$-qubit random quantum circuit, as defined in \autoref{def:rqcs}, of depth $144n$ passes the LXEB test in \autoref{prob:lxeb} with $b=1.97$, with probability over the choice of random circuit $U\sim \nu$ and over $k$ output samples $x_1,\ldots,x_k\sim p_U(x)$ is lower bounded as
    \begin{equation}
        \Pr_{\substack{U\sim \nu\\ x_1,\ldots,x_k\sim p_U}}\Bigg( \frac{1}{k} \sum_{i=1}^k p_U(x_i) \geq \frac{1.97}{2^n}\Bigg) \geq 1 - \frac{200 \sqrt{2}}{\sqrt{k}} - \frac{50000}{2^n}\,.
    \end{equation}
\end{theorem}

To prove the theorem, we start by showing that the collision probability for a 4-design is highly concentrated about the Haar averaged value $\Ex_{U\sim\mu_H}[\Ex_{x\sim p_U}[p_U(x)]] = 2/(d+1)$.
A similar inequality was proven in Ref.~\cite{nietner2023average} to show that output distributions of random quantum circuits are far from uniform with high probability. 
\begin{proposition}\label{prop:cpconc}
    Let $\nu$ be an $1/d^5$-approximate unitary $4$-design. The probability that the collision probability deviates from $2/(d+1)$ is bounded as
    \begin{equation}
        \Pr_{U\sim \nu}\Big( \Big|\Ex_{x\sim p_U}\big[ p_U(x)\big] - \frac{2}{d+1}\Big|\geq \delta\Big)\leq \frac{5}{d^3\delta^2}\,.
    \end{equation}
\end{proposition}
\begin{proof}
    First, consider the variance of the collision probability with respect to the Haar measure. Recall that the Haar averaged collision probability is $\Ex_{U\sim \mu_H}[\Ex_x[ p_U(x)]] = \frac{2}{d+1}$. The variance is then
    \begin{align}
    \begin{split}
        \Var_{U\sim \mu_H} \Big( \Ex_x\big[ p_U(x)\big] \Big) &= \Ex_U \bigg[ \bigg(\sum_x p_U(x)^2\bigg)^2 \bigg] - \Ex_U\bigg[\sum_xp_U(x)^2\bigg]^2 \\
        &= \sum_{x,y} \Ex_U[p_U(x)^2p_U(y)^2] - \bigg(\sum_x\Ex_U[p_U(x)^2]\bigg)^2\\
        &= \sum_{x} \Ex_U[p_U(x)^4] + \sum_{x\neq y} \Ex_U[p_U(x)^2p_U(y)^2] -\left(\frac{2}{(d+1)}\right)^2\\
        &= d\frac{24}{d(d+1)(d+2)(d+3)} + d(d-1)\frac{4}{d(d+1)(d+2)(d+3)}-\frac{4}{(d+1)^2}\\
        &=\frac{4(d-1)}{(d+1)^2(d+2)(d+3)}\,.
        \end{split}
    \end{align}
    Now if $\nu$ is an $1/d^5$-approximate unitary $4$-design, we have
    \begin{equation}
        \Ex_{U\sim \nu}[p_U(x)^2p_U(y)^2] \leq \Ex_{U\sim \mu_H}[p_U(x)^2p_U(y)^2] + \frac{1}{d^5} \quad{\rm and}\quad \Ex_{U\sim \nu}[p_U(x)^2] \geq \Ex_{U\sim \mu_H}[p_U(x)^2] - \frac{1}{d^5}
    \end{equation}
    and thus centering the collision probability on the Haar value and looking at the second moment, we find
    \begin{align}
        \Ex_{U\sim \nu}\bigg[ \bigg|\Ex_x\big[ p_U(x)\big] - \frac{2}{d+1}\bigg|^2\bigg] 
        &= \sum_{x,y} \Ex_{U\sim \nu}[p_U(x)^2p_U(y)^2] -2 \left(\frac{2}{d+1}\right)\sum_{x} \Ex_{U\sim \nu}[p_U(x)^2] + \left(\frac{2}{d+1}\right)^2\nn
        &\leq \Var_{U\sim \mu_H} \Big( \Ex_x\big[ p_U(x)\big] \Big) + \frac{d^2}{d^5} + \frac{4}{d^5}\nn
        &\leq \frac{5}{d^3}\,.
    \end{align}
    It thus follows from Markov's inequality that the collision probability is exponentially concentrated around $2/(d+1)$ over the choice over unitary from an approximate $4$-design as
    \begin{equation}
        \Pr_{U\sim \nu}\Big( \Big|\Ex_x\big[ p_U(x)\big] - \frac{2}{d+1}\Big|\geq \delta\Big)\leq \frac{1}{\delta^2} \Ex_{U\sim \nu}\bigg[ \bigg|\Ex_x\big[ p_U(x)\big] - \frac{2}{d+1}\bigg|^2\bigg]\leq \frac{5}{d^3\delta^2}\,,
    \end{equation}
which concludes the proof.
\end{proof}

We now prove a bound on the variance of the sample mean of $p_U(x)$'s:
\begin{lemma}\label{lem:meanvar}
    Let $\nu$ be an $1/d^5$-approximate unitary $4$-design, then over the choice of $U\sim \nu$, the variance over $x_i$'s of the sample mean of the output probabilities $p_U(x_i)$ obeys
    \begin{equation}
        \Pr_{U\sim\nu}\bigg( \Var_{x_i\sim p_U}\bigg(\frac{1}{k}\sum_{i=1}^k p_U(x_i)\bigg)\geq \delta'\bigg) \leq \frac{2}{\delta'kd^2}\,.
    \end{equation}
\end{lemma}
\begin{proof}
    First, recall that for iid random variables, the variance of the sample mean is simply expressed in terms of the variance of the samples
    \begin{equation}
        \Var_{x_i\sim p_U}\bigg(\frac{1}{k}\sum_{i=1}^k p_U(x_i)\bigg) = \frac{1}{k} \Var_{x\sim p_U} \big(p_U(x)\big)\,.
    \end{equation}
    The variance of the output probabilities over the samples can be expressed as
    \begin{equation}
        \Var_{x\sim p_U}\big(p_U(x)\big) = \Ex_{x}\big[p_U(x)^2\big] - \Ex_{x}\big[p_U(x)\big]^2 = \sum_x p_U(x)^3 - \Big(\sum_x p_U(x)^2\Big)^2\,.
    \end{equation}
    It is then straightforward to compute the Haar expectation of the variance using \autoref{lem:pUmoms} for the moments of output probabilities
    \begin{align}
    \begin{split}
        \Ex_{U\sim\mu_H} \Big[\Var_{x\sim p_U} \big(p_U(x)\big)\Big] &= \sum_x \Ex_U[p_U(x)^3] - \sum_x \Ex_U[p_U(x)^4] - \sum_{x\neq y} \Ex_U[p_U(x)^2p_U(y)^2]\\
        &=\frac{2(d-1)}{(d+3)(d+2)(d+1)}\,.
    \end{split}
    \end{align}
    For an $1/d^5$-approximate unitary $4$-design, it follows from \autoref{lem:approxp} that for $r\leq 4$ the additive error for each $\Ex_U[p_U(x)^r]$ is at most $1/d^5$. Therefore the expectation of the output probability variance over an approximate 4-design is
    \begin{equation}
        \Ex_{U\sim\nu} \Big[\Var_{x\sim p_U} \big(p_U(x)\big)\Big] = \frac{2(d-1)}{(d+3)(d+2)(d+1)} + \frac{d}{d^5} + \frac{d^2}{d^5}\leq \frac{2}{d^2}\,.
    \end{equation}
    The lemma now follows from Markov's inequality
    \begin{equation}
        \Pr_{U\sim\nu}\bigg( \Var_{x_i\sim p_U}\bigg(\frac{1}{k}\sum_{i=1}^k p_U(x_i)\bigg)\geq \delta'\bigg) \leq \frac{1}{\delta'k} \Ex_{U\sim\nu} \Big[\Var_{x\sim p_U} \big(p_U(x)\big)\Big]\leq \frac{2}{\delta'kd^2}\,.
    \end{equation}
    
\end{proof}

We can now use \autoref{prop:cpconc} and \autoref{lem:meanvar} to prove \autoref{thm:lindepthfrml}. The proof follows from two observations, the sample mean of the output probabilities $p_U(x_i)$ concentrates around the collision probability with high probability over the samples and the choice of unitary drawn from an approximate $4$-design. Moreover, for a unitary 4-design, the collision probability is exponentially well concentrated around the Haar value $\frac{2}{d+1}$.

\begin{proof}[Proof of \autoref{thm:lindepthfrml}]
First, we know from \autoref{prop:cpconc} and \autoref{lem:meanvar} that for a unitary $U$ drawn from an $1/d^5$-approximate unitary 4-design
\begin{equation}
    \Pr_{U\sim \nu}\Big( \Big|\Ex_{x\sim p_U}\big[ p_U(x)\big] - \frac{2}{d+1}\Big|\geq \delta\Big)\leq \frac{5}{d^3\delta^2}
    \quad{\rm and}\quad 
    \Pr_{U\sim\nu}\bigg( \Var_{x_i\sim p_U}\bigg(\frac{1}{k}\sum_{i=1}^k p_U(x_i)\bigg)\geq \delta'\bigg) \leq \frac{2}{\delta'kd^2}\,.
\end{equation}
Using Chebyshev's inequality over $x_1,\ldots,x_k\sim p_U(x)$, we know that
\begin{equation}
    \Pr_{x_i}\left( \bigg| \frac{1}{k}\sum_{i=1}^k p_U(x_i) - \frac{1}{k}\sum_{i=1}^k \mathbb{E}_{x_i} [p_U(x_i)]\bigg| \geq \delta'' \right)\leq \frac{{\rm Var}_{x_i}\bigg(\frac{1}{k}\sum_{i=1}^k p_U(x_i)\bigg)}{\delta''^2}\,.
\end{equation}
Therefore, it follows from a union bound that with high probability over both the choice of unitary from a $1/d^5$-approximate unitary 4-design and the output samples the emperical mean of $p_U(x_i)$'s concentrates around the collision probability, which concentrates around the Haar averaged collision probability. Thus, from a union bound, we find
\begin{equation}
    \Pr_{\substack{U\sim \nu\\ x_1,\ldots,x_k\sim p_U}}\Bigg( \frac{1}{k}\sum_{i=1}^k p_U(x_i)\geq \frac{2}{d+1}-\delta-\delta''\bigg) \geq 1 - \frac{2}{\delta'd^2 k} - \frac{\delta'}{\delta''^2} - \frac{5}{d^3\delta^2}
\end{equation}
To satisfy the condition in \autoref{prob:lxeb} and pass the LXEB test with $b=1.97$, we must take $\delta=\delta''=1/(100 d)$ and make the mild assumption that $n\geq 8$. This choice of $\delta$ and $\delta''$ demands that we take $\delta'=c/\sqrt{k}d^2$. The theorem for unitary design elements then follows for the optimal choice of $c$ and thus gives that we pass the LXEB test with probability $\geq 1 - O(1/\sqrt{k})-O(1/d)$.

Finally, it follows from \autoref{thm:4designs} that random quantum circuits form $1/d^5$-approximate unitary 4-designs in a depth $144n$. 
\end{proof}

It seems worthwhile to point out that the concentration we obtain for the random variable $\mathbb{E}_x [p_U(x)]$ is stronger than we would expect from Levy's lemma for Haar random unitaries.
Indeed, Levy's lemma only implies a meaningful concentration inequality for a radius of $\Omega(2^{n/2})$ around the expectation value as the Lipschitz constant of $\mathbb{E}_x [p_U(x)]$ is constant.
However, this is not strong enough to imply the type of exponential concentration we find even relative to the scaling of the expectation value.

\subsection{Improved scaling for linear depth random quantum circuits}
The bound for approximate unitary 8-designs will be a convenient intermediate step for the stronger $n$-design bound in the next subsection.
Here, we prove a stronger convergence result than \autoref{thm:lindepth} for random quantum circuits forming an approximate $8$-design.
\begin{proposition}\label{prop:8design}
    A random quantum circuit of depth $O(n)$ will pass the LXEB test in \autoref{prob:lxeb} with $b=1.97$, with probability over the choice of random circuit $U\sim \nu$ and over output samples $x_1,\ldots,x_k\sim p_U$ that we pass lower bounded as
    \begin{equation}
        \Pr_{\substack{U\sim \nu\\ x_1,\ldots,x_k\sim p_U}}\Bigg( \frac{1}{k} \sum_{i=1}^k p_U(x_i) \geq \frac{1.97}{2^n}\Bigg) \geq 1 - \frac{400}{k} - \frac{210000}{2^n}\,.
    \end{equation}
\end{proposition}
Although random quantum circuits are known to form 8-designs (or any constant design order) in linear depth, we insist on using 4-designs for our main theorem as the constants in the design depth are known to be manageable. Exact and numerical calculation of the spectral gaps for the moment operators of random circuits \cite{BH13,HHJ20,MHJ23} up to the $6$th moment, guarantees good constants in the design depth $c(n t + \log_2(1/\ep))$, \eg $c<100$. Beyond 6th moments, we must resort to the stronger results for arbitrary moments \cite{BHH12,haferkamp2022random,chen2025incompressibility}, where the constants involved are large, \eg $10^{16}$ for the 8th moment.

To get stronger convergence in the probability that we pass the LXEB test, we can make use of 8th moments. We have already shown that the expected value over unitaries of the variance over samples is bounded as $\Ex_{U\sim \nu}[\Var_{x_i}(x_i)]\leq 2/d^2$. To get a probabilistic bound on the variance over samples to upper bound the variance with high probability, we can compute the variance of the variance.

\begin{lemma}\label{lem:varpU}
Let $\nu$ be an $1/d^9$-approximate unitary $8$-design, with high probability the variance over samples concentrates as
\begin{equation}
\Pr_{U\sim\nu}\left(\Var_{x\sim p_U}\big(p_U(x)\big)\geq \frac{8}{d^2}\right) \leq \frac{4}{d}.
\end{equation}
\end{lemma}

\begin{proof}
    First we compute the variance over unitaries drawn from the Haar measure of the variance over samples from the output distribution. Recall that
    \begin{equation}
        \Var_{x\sim p_U}\big(p_U(x)\big) = \Ex_{x}\big[p_U(x)^2\big] - \Ex_{x}\big[p_U(x)\big]^2 = \sum_x p_U(x)^3 - \Big(\sum_x p_U(x)^2\Big)^2\,.
    \end{equation}
    Analogous to the calculation of the mean, we square the variance and compute the Haar expectation
    \begin{align}
        &\Ex_{U\sim\mu_H}\!\! \bigg[\bigg(\!\Var_{x\sim p_U}\!\!\big(p_U(x)\big)\!\bigg)^2\bigg]\\
        &\quad = \sum_{x,y} \Ex_U[p_U(x)^3 p_U(y)^3] - 2\sum_{x,y,x} \Ex_U[p_U(x)^3 p_U(y)^2 p_U(z)^2] +\!\! \sum_{w,x,y,z} \Ex_U[p_U(w)^2 p_U(x)^2 p_U(y)^2 p_U(z)^2]\,,\nonumber
    \end{align}
    where all expectations are taken over the Haar measure. For each of the three terms, we need to separately treat the collisions in the sum, \ie considering when the basis states in the sums are equal or not. For the first term, we find two terms
    \begin{align}
    \begin{split}
        \sum_{x,y} \Ex_U[p_U(x)^3 p_U(y)^3] &= d \Ex_U[p_U(x)^6] + d(d-1) \Ex_U[p_U(x)^3p_U(y)^3]\\
        &= \frac{36 (d+19)}{(d+1) (d+2) (d+3) (d+4) (d+5)}\,,
    \end{split}
    \end{align}
    where $x\neq y$ on right-hand side in the first equality. For the second term, we find
    \begin{align}
    \begin{split}
        &\sum_{x,y,z} \Ex_U[p_U(x)^3 p_U(y)^2 p_U(z)^2]\\ &\quad = d \Ex_U[p_U(x)^7] + d(d-1) \Ex_U[p_U(x)^3 p_U(y)^4] + 2d(d-1)\Ex_U[p_U(x)^5 p_U(y)^2]\\
        &\quad\qquad + d(d-1)(d-2)\Ex_U[p_U(x)^3 p_U(y)^2 p_U(z)^2]\\
        &\quad =\frac{24 (d^2+23 d+186)}{(d+1) (d+2) (d+3) (d+4) (d+5) (d+6)}\,.
        \end{split}
    \end{align}
    Lastly, in the third term we find
    \begin{align}
    \begin{split}
        &\sum_{w,x,y,z} \Ex_U[p_U(w)^2 p_U(x)^2 p_U(y)^2 p_U(z)^2]\\ &\quad = d\Ex_U[p_U(x)^8]+ 6d(d-1)\Ex_U[p_U(x)^4p_U(y)^2p_U(z)^2] + 3d(d-1)\Ex_U[p_U(x)^4p_U(y)^4]\\
        &\quad\qquad + 4d(d-1)\Ex_U[p_U(x)^6p_U(y)^2] + d(d-1)(d-2)(d-3)\Ex_U[p_U(w)^2p_U(x)^2p_U(y)^2p_U(z)^2]\\
        &\quad= \frac{16 \left(d^3+30 d^2+371 d+2118\right)}{(d+1) (d+2) (d+3) (d+4) (d+5) (d+6) (d+7)}\,.
        \end{split}
    \end{align}
    Again, we treat the collisions in the sum separately, such that the arguments of the $p_U(x)$'s in the expectations in the second line are not equal. 
    Given the exact expression for the mean of the variance $\Ex_{U\sim\mu_H}[\Var_{x\sim p_U}(p_U(x))]$ in \autoref{lem:meanvar}, direct computation gives 
    \begin{align}
    \begin{split}
        \Var_{U\sim\mu_H} \bigg[\Var_{x\sim p_U}\big(p_U(x)\big)\bigg] &= \Ex_{U\sim\mu_H} \bigg[\bigg(\Var_{x\sim p_U}\big(p_U(x)\big)\bigg)^2\bigg] - \Ex_{U\sim\mu_H} \bigg[\Var_{x\sim p_U}\big(p_U(x)\big)\bigg]^2\\
        &= \frac{8 \left(17 d^5+42 d^4-106 d^3-72 d^2+449 d-330\right)}{(d+1)^2 (d+2)^2 (d+3)^2(d+4) (d+5) (d+6) (d+7)}\\
        &\leq \frac{136}{d^5}\,.
        \end{split}
    \end{align}
    Now for $1/d^9$-approximate unitary 8-designs, we can use \autoref{lem:approxp} to bound the additive error of the moments of the output probabilities $\Ex_{U\sim \nu}[p_U(x)^r]$ for any $r\leq 8$ to be at most $1/d^9$. It thus follows that
    \begin{align}
        &\Ex_{U\sim\nu}\bigg[\Big(\Var_{x\sim p_U}\big(p_U(x)\big)\Big)^2\bigg] \leq \Ex_{U\sim\mu_H}\bigg[\Big(\Var_{x\sim p_U}\big(p_U(x)\big)\Big)^2\bigg] + \frac{d^2}{d^9} + 2\frac{d^3}{d^9} + \frac{d^4}{d^9}\\[4pt]
        &\Ex_{U\sim\nu}\bigg[\Var_{x\sim p_U}\big(p_U(x)\big)\bigg] \geq \Ex_{U\sim\mu_H}\bigg[\Var_{x\sim p_U}\big(p_U(x)\big)\bigg] - \frac{d}{d^9} - \frac{d^2}{d^9}\,.
    \end{align}
    We then center the variance around the Haar expectation and look at the second moment with respect to our unitary 8-design    
    \begin{align*}
        &\Ex_{U\sim\nu}\left[\bigg|\Var_{x\sim p_U}\big(p_U(x)\big)-\Ex_{U\sim\mu_H}\bigg[\Var_{x\sim p_U}\big(p_U(x)\big)\bigg]\bigg|^2\right]\\
        &\qquad =  \Ex_{U\sim\nu}\bigg[\Big(\Var_{x\sim p_U}\big(p_U(x)\big)\Big)^2\bigg] - 2\Ex_{U\sim\nu}\bigg[\Var_{x\sim p_U}\big(p_U(x)\big)\bigg] \Ex_{U\sim\mu_H}\bigg[\Var_{x\sim p_U}\big(p_U(x)\big)\bigg] + \Ex_{U\sim\mu_H}\bigg[\Var_{x\sim p_U}\big(p_U(x)\big)\bigg]^2\\
        &\qquad \leq \Var_{U\sim\mu_H} \bigg[\Var_{x\sim p_U}\big(p_U(x)\big)\bigg] + \frac{4}{d^5} + \frac{4}{d^7}\\
        &\qquad \leq\frac{144}{d^5}\,.
    \end{align*}
    Now noting that
    \begin{equation}
        \Ex_{U\sim\mu_H}\bigg[\Var_{x\sim p_U}\big(p_U(x)\big)\bigg] = \frac{2(d-1)}{(d+1)(d+2)(d+3)}\leq \frac{2}{d^2}\,,
    \end{equation}
    we find
    \begin{align}
    \begin{split}
        \Pr_{U\sim\nu}\left(\Var_{x\sim p_U}\big(p_U(x)\big)\geq \frac{8}{d^2}\right) &\leq \Pr_{U\sim\nu}\left(\bigg|\Var_{x\sim p_U}\big(p_U(x)\big)-\Ex_{U\sim\mu_H}\bigg[\Var_{x\sim p_U}\big(p_U(x)\big)\bigg]\bigg|^2\geq \frac{36}{d^4}\right)\\
        &\leq \frac{d^4}{36} \Ex_{U\sim\nu}\left[\bigg|\Var_{x\sim p_U}\big(p_U(x)\big)-\Ex_{U\sim\mu_H}\bigg[\Var_{x\sim p_U}\big(p_U(x)\big)\bigg]\bigg|^2\right]\\
        &\leq \frac{4}{d}\,,
        \end{split}
    \end{align}
    which completes the proof.
\end{proof}
We now turn to showing the improved scaling for linear depth random quantum circuits from approximate unitary 8-designs.
\begin{proof}[Proof of \autoref{prop:8design}]
    From \autoref{prop:cpconc}, we know that $\Ex_{x\sim p_U}[p_U(x)]$ concentrates very sharply around the Haar value $2/(d+1)$. Now \autoref{lem:varpU} gives a stronger bound on the sample variance than in \autoref{thm:lindepthfrml}. Specifically, using Chebyshev's inequality just as before to show concentration of the sample mean, we can then bound the variance over output samples using \autoref{lem:varpU}. It follows from \autoref{lem:varpU} that
    \begin{equation}
        \Pr_{U\sim\nu}\left(\Var_{x\sim p_U}\big(p_U(x)\big)\geq \frac{8}{c d^2}\right) \leq \frac{4c^2}{d}\,.
    \end{equation}
    Therefore, a union bound then implies that with high probability over both the choice of unitary from a $1/d^9$-approximate unitary 8-design and the output samples the sample mean of $p_U(x_i)$'s concentrates as
    \begin{equation}
        \Pr_{\substack{U\sim \nu\\ x_1,\ldots,x_k\sim p_U}}\Bigg( \frac{1}{k}\sum_{i=1}^k p_U(x_i)\geq \frac{2}{d+1}-\delta-\delta'\bigg) \geq 1 - \frac{4 c^2}{d} - \frac{8}{k c d^2 \delta'^2} - \frac{5}{d^3\delta^2}\,.
    \end{equation}
    Taking $\delta=\delta'=1/(100d)$, making a convenient choice of $c$, and again assuming that we have $n\geq 8$ qubits, we find we pass the LXEB test with $b=1.97$, and with probability $\geq 1 - O(1/k) - O(1/d)$.

    Lastly, it follows from \autoref{thm:tdesigns} that random quantum circuits form $1/d^9$-approximate unitary 8-designs in linear depth. 
\end{proof}

\subsection{Polynomial depth random quantum circuits pass the LXEB test}
We now prove our main result, that polynomial depth random quantum circuits pass the linear cross-entropy test with probability greater than $1-O(e^{-k\log(n)/n})-O(1/2^n)$. The formal version of \autoref{thm:polydepth} is stated as follows.
\begin{theorem}[Restatement of \autoref{thm:polydepth}]\label{thm:polydepthfrml}
    An $n$-qubit random quantum circuit, as defined in \autoref{def:rqcs}, of depth $O(n^2\, \mathrm{polylog}(n))$ passes the LXEB test in \autoref{prob:lxeb} with $b=1.97$, with probability over the choice of random circuit $U\sim \nu$ and over $k$ output samples $x_1,\ldots,x_k\sim p_U$ lower bounded as
    \begin{equation}
        \Pr_{\substack{U\sim \nu\\ x_1,\ldots,x_k\sim p_U}}\Bigg( \frac{1}{k}\sum_{i=1}^k p_U(x_i)\geq \frac{1.97}{2^n}\bigg) \geq 1 - \exp\left( - \frac{k}{800n}\log\left(\frac{n}{200}\right)\right) - \frac{50006}{2^n}\,.
    \end{equation}
\end{theorem}

The proof of this follows from a bound on the maximum output probability, the strong concentration of the collision probability about its Haar value, and some concentration inequalities for sums of random variables. 

First, for approximate $n$-designs, and thus polynomial depth random quantum circuits, with high probability the max output probability is $O(n/2^n)$.

\begin{theorem}[Restatement of \autoref{thm:maxpUinf}]\label{thm:maxpU}
    Let $\nu$ be an $1/d^{n}$-approximate unitary $n$-design, over the choice of unitaries $U\sim \nu$ the largest output probability $\max_x p_U(x)$ obeys
    \begin{equation}
        \Pr_{U\sim\nu}\left( \max_x p_U(x) \geq \frac{4n}{d}\right) \leq \frac{2}{d}\,.
    \end{equation}
    Therefore, with probability exponentially close to 1, the largest output probability is $O(n/d)$.
\end{theorem}
We defer the proof of this claim until the end of the section, but a key observation is that the $n$-th moment `sees' the infinity norm. Note that
\begin{equation}
    \Ex_U\big[\max_x p_U(x)\big] = \Ex_U\big[\|p_U(x)\|_{\ell_\infty}\big] \leq \Ex_U\big[\|p_U(x)\|_{\ell_q}\big] \leq \bigg(\Ex_U\Big[\sum_x p_U(x)^q\Big]\bigg)^{1/q} \leq \frac{d^{1/q} q}{d}\,,
\end{equation}
using the monotonicity of vector norms and Jensen's inequality. Taking $q=n=\log(d)$, we observe that $\Ex_U[\max_x p_U(x)] \leq 2n/d$. To show this occurs for approximate design elements with probability exponentially close to one takes a few more steps. This behavior cannot be asymptotically improved as it is the correct scaling for the max output probability of Haar random unitaries \cite{aaronson2023certified}.

Given the concentration of the empirical mean about the collision probability and our bound on the max output probability $p_U(x)$, the statement that random circuits of polynomial depth pass the LXEB test with probability $1-e^{-k/n^2}$ follows from Hoeffding's inequality as
\begin{equation}
    \Pr \bigg( \bigg|\frac{1}{k} \sum_{i=1}^k p_U(x_i) - \Ex_{x_i}[p_U(x_i)]\bigg| \geq \frac{\delta}{d}\bigg) \leq 2 e^{-2k\delta^2/16n^2}\,.
\end{equation}
But we can improve on Hoeffding by using Bennett's inequality and the concentration of the variance of the output probabilities.

\begin{proof}[Proof of \autoref{thm:polydepthfrml}]
    Given the bound on the maximum output probability $\max_x p_U(x) \leq 4n/d$ with probability exponentially close to one, we can use standard concentration inequalities for sums of independent random variables. Bennett's inequality \cite{Bennett62,concineq}, similar Bernstein's inequality, says that for random variables $X_1,\ldots,X_k$ with $\Ex[X_i]=\mu$, finite variance $\sigma^2$, and $X_i\leq \alpha$ almost surely, then for any $\delta>0$
    \begin{equation}
        \Pr \bigg( \,\bigg|\frac{1}{k} \sum_{i=1}^k X_i - \mu \bigg| \geq \delta\bigg) \leq \exp\left( - \frac{k \sigma^2}{\alpha^2} h\left( \frac{\alpha\delta}{\sigma^2}\right)\right)\,,
    \end{equation}
    where $h(u) := (1+u)\log(1+u)-u$. Returning to the sample mean of output probabilities, we know from \autoref{lem:varpU} and \autoref{thm:maxpU} that over the choice unitary $U\sim\nu$ drawn from an approximate $\max\{n,8\}$-design, we have
    \begin{equation}
        \Pr_{U\sim\nu}\left(\Var_{x\sim p_U}\big(p_U(x)\big)\geq \frac{8}{d^2}\right) \leq \frac{4}{d} \quad{\rm and}\quad  \Pr_{U\sim\nu}\left( \max_x p_U(x) \geq \frac{4n}{d}\right) \leq \frac{2}{d}\,.
    \end{equation}
    Taking a union bound and noting that $h(u)\geq u \log(u)/2$ for $u\geq 0$, we can then apply Bennett's inequality and find
    \begin{equation}
        \Pr_{x_1,\ldots,x_k\sim p_U} \bigg( \,\bigg|\frac{1}{k} \sum_{i=1}^k p_U(x_i) - \Ex_{x_i}[p_U(x_i)] \bigg| \geq \frac{\delta}{d}\bigg) \leq \exp\left( - \frac{k\delta}{8n}\log\left(\frac{n\delta}{2}\right)\right)\,.
    \end{equation}
    Using the concentration of the expected output probability $\Ex_{x\sim p_U}[p_U(x)]$ about the Haar averaged collision probability in \autoref{prop:cpconc} and the concentration of the sample mean from Bennett's inequality, we again take a union bound and find
    \begin{equation}
        \Pr_{\substack{U\sim \nu\\ x_1,\ldots,x_k\sim p_U}}\Bigg( \frac{1}{k}\sum_{i=1}^k p_U(x_i)\geq \frac{1.97}{d}\Bigg) \geq 1 - \exp\left( - \frac{k}{800n}\log\left(\frac{n}{200}\right)\right) - \frac{50006}{d}\,,
    \end{equation}
    and thus over the choice of unitaries drawn from a $1/d^{n'}$-approximate unitary $n'$-design, where $n'=\max\{n,8\}$, and over the output samples, we pass the LXEB test with probability $\geq 1-O(e^{-k \log(n)/n})-O(1/d)$. 

    Given the result in Ref.~\cite{chen2025incompressibility}, restated in \autoref{thm:tdesigns}, it follows that $O(n^{2}\,\mathrm{polylog}(n))$ depth random circuits form $1/d^n$-approximate $n$-designs, the theorem follows.
\end{proof}

Now we prove that the max output probability of a unitary drawn from an approximate $n$-design is $O(n/d)$ with high probability.
\begin{proof}[Proof of \autoref{thm:maxpU}]
    First, let $\nu$ be an $\ep$-approximate unitary $t$-design. We can bound how large a single output probability is using higher moments and Markov's inequality as follows
    \begin{equation}
        \Pr_{U\sim \nu} \bigg( p_U(x) \geq \frac{4n}{d}\bigg) = \Pr_{U\sim \nu} \bigg( (p_U(x))^t \geq \frac{(4n)^t}{d^t}\bigg) \leq \frac{d^t}{(4n)^t} \Ex_{U\sim \nu} \big[(p_U(x))^t\big]\,.
    \end{equation}
    Taking $\ep=1/d^t$, we know from \autoref{lem:pUmoms} that for an $1/d^t$-approximate $t$-design the moments of the output probabilities are bounded as
    \begin{equation}
        \Ex_{U\sim \nu} \big[(p_U(x))^t\big] \leq \frac{t!}{\prod_{i=0}^{t-1}{(d+i)}}+\frac{1}{d^t}\leq \frac{t!+1}{d^t}\leq 2\frac{t!}{d^t}\,.
    \end{equation}
    Therefore, taking a union bound over all $2^n$ output probabilities and taking $t=n$ we find
    \begin{equation}
        \Pr_{U\sim \nu} \bigg( \max_x p_U(x) \geq \frac{4n}{d}\bigg) \leq 2^n\frac{2(n!)}{(4n)^n}\leq \frac{2}{d}\,.
    \end{equation}
\end{proof}

We also note that one can prove the same result, albeit for approximate $2n$-designs, without taking the union bound and with a little more work. In some sense, this approach more directly shows that the $\log(d)$-th norm sees the infinity norm. 

We first note that the variance of the $q$-th moments of the probabilites is
\begin{align}
    \Var_{U\sim\mu_H}\bigg(\sum_x p_U(x)^q\bigg) &= \Ex_U\bigg[\Big(\sum_x p_U(x)^q\Big)^2\bigg] - \Ex_U\Big[ \sum_x p_U(x)^q\Big]^2\nn
    &= \sum_x \Ex_U\big[p_U(x)^{2q}\big] + \sum_{x,y} \Ex_U\big[p_U(x)^q p_U(y)^q\big] - \bigg(\sum_x \Ex_U\big[p_U(x)^q\big]\bigg)^2\nn
    &= \frac{d(2q)!}{d(d+1)\ldots(d+2q-1)} + \frac{d(d-1)(q!)^2}{d(d+1)\ldots(d+2q-1)} - \frac{d^2(q!)^2}{(d(d+1)\ldots(d+q-1))^2}\nn
    &= \frac{(2q)!}{(d+1)\ldots(d+2q-1)} - \frac{(q!)^2}{(d(d+1)\ldots(d+q-1))^2}\nn
    &\leq \frac{(2q)!}{d^{2q-1}}\,.
\end{align}
Let $\nu$ be an $1/d^{2q+1}$-approximate unitary $2n$-design, then for $\Var_{U\sim\nu}\big(\sum_x p_U(x)^q\big) \leq \frac{(2q)!+2}{d^{2q-1}} $. As we noted earlier, the expected value of the largest output probability is bounded as $n/d$
\begin{equation}
    \Ex_U\big[\max_x p_U(x)\big] = \Ex_U\big[\|p_U(x)\|_{\ell_\infty}\big] \leq \Ex_U\big[\|p_U(x)\|_{\ell_q}\big] \leq \bigg(\Ex_U\Big[\sum_x p_U(x)^q\Big]\bigg)^{1/q} \leq \frac{d^{1/q} q}{d}\,,
\end{equation}
where we use the monotonicity of vector norms and Jensen's inequality, and in the final inequality we have used
\begin{equation}
    \Ex_U\Big[\sum_x p_U(x)^q\Big] = d\Ex_U\big[p_U(x)^q\big] = \frac{q!}{(d+1)\ldots(d+q-1)} \leq \frac{q^q}{d^{q-1}}\,.
\end{equation}
The same bound holds for $\nu$ as $q!+1\leq q^q$ as long as $q\geq 2$. Therefore, we have
\begin{align}
    \Pr_{U\sim\nu}\left( \max_x p_U(x) \geq \frac{2 d^{1/q} q}{d}\right) &\leq \Pr_{U\sim\nu}\left( \| p_U(x)\|_{\ell_q} \geq \frac{2 d^{1/q} q}{d}\right) = \Pr_{U\sim\nu}\bigg( \sum_x p_U(x)^q \geq \frac{(2q)^q}{d^{q-1}}\bigg)\nn
    &\leq \Pr_{U\sim\nu}\bigg( \bigg|\sum_x p_U(x)^q-\Ex_{U\sim\nu}\Big[\sum_x p_U(x)^q\Big]\bigg| \geq \frac{(2q)^q}{d^{q-1}}-\Ex_{U\sim\nu}\Big[\sum_x p_U(x)^q\Big]\bigg)\nn
    &\leq \Pr_{U\sim\nu}\bigg( \bigg|\sum_x p_U(x)^q-\Ex_{U\sim\nu}\Big[\sum_x p_U(x)^q\Big]\bigg| \geq \frac{(2^q-1) q^q}{d^{q-1}}\bigg)\nn
    &\leq \frac{(2q)^{2q}}{d^{2q-1}} \frac{d^{2q-2}}{(2^q-1)^2 q^{2q}}\nn
    &\leq \frac{2}{d}\,,
\end{align}
where $2^{2q}/(2^q-1)^2\leq 2$ under the mild assumption that $q\geq 2$. Taking $q=n$, we find $\Pr_{U\sim\nu}( \max_x p_U(x) \geq 4n/d)$ just as in \autoref{thm:maxpU}, but for approximate $2n$-designs.

\vspace*{-8pt}
\subsection*{Acknowledgments}
\vspace*{-4pt}
We would like to acknowledge helpful discussions with Soumik Ghosh, Markus Heinrich, Shih-Han Hung, Ingo Roth, and Sanjay Shakkottai. We would especially like to thank Scott Aaronson for suggesting this problem to us and for valuable comments. JH and NHJ would like to thank the Centro de Ciencias de Benasque Pedro Pascual and the Simons Institute for the Theory of Computing for fostering a productive environment, which assisted in the completion of this work. NHJ acknowledges support in part from DOE grant DE-SC0025615. JH acknowledges funding from the Harvard Quantum Initiative.

\appendix

\section{LXEB test for other ensembles}\label{sec:otherensembles}
In this appendix, we consider a few other circuit ensembles and establish similar results to our main theorems.

\subsection*{Orthogonal designs}
We will denote the $d$-dimensional orthogonal group as $\O(d)$ and its Haar measure as $\mu_\O$. For $\O\sim \mu_\O$ drawn randomly from the Haar measure on the orthogonal group, the analogous expression to \autoref{lem:pUmoms} for the moments of the output probabilities is as follows.
\begin{lemma}\label{lem:pOmoms}
For any positive integer $t$ and for $\O\sim \mu_\O$ drawn from the Haar measure on the orthogonal group $\O(d)$, the moments of the output probabilities $p_\O(x) = |\vev{x|\O|0}|^2$ are given by
\begin{equation}
    \Ex_{\O\sim \mu_\O}\bigg[\prod_{i=1}^\ell |p_\O(x_i)|^{\lambda_i}\bigg] = \frac{\prod_{i=1}^\ell \frac{(2\lambda_i)!}{2^{\lambda_i} \lambda_i!}}{\prod_{i=0}^{t-1} (d+2i)}\,,
\end{equation}
where all bit-strings $x_i$ are distinct and where $\lambda\vdash t$ is an integer partition of $t$ such that $\lambda = (\lambda_1,\ldots,\lambda_\ell)$ and $\sum_{i=1}^\ell \lambda_i = t$. 
\end{lemma}

\ni We note that in \autoref{lem:pUmoms} the initial state $\ket{0}$ is arbitrary, but above it is important that $\ket{0}$ is real valued.
\begin{proof}[Proof of \autoref{lem:pOmoms}]
The $t$-th moments of random real states $\ket{\vphi}=\O\ket{0}$ with $\O\sim \mu_\O$ are given as
\begin{equation}\label{eq:Ostates}
    \Ex[\ketbra{\vphi}^{\otimes t}] = \frac{1}{{\prod_{i=0}^{t-1} (d+2i)}} \sum_{\sigma \in M_{2t}} \Pi_\sigma
\end{equation}
where $M_{2t}$ is the set of all pair partitions and $\Pi_\sigma$ is an operator acting on the $t$-fold Hilbert space. A pair partition (or perfect matching) $\sigma\in M_{2t}$ is a partition of $\{1,\ldots, 2t\}$ into pairs, written as $\{\{\sigma(1), \sigma(2)\},\ldots,\{\sigma(2t-1), \sigma(2t)\}\}$, where $\sigma(2i-1)< \sigma(2i)$ and $\sigma(1)<\sigma(3)<\ldots<\sigma(2t-1)$. The operator $\Pi_\sigma$ is then defined as $\Pi_\sigma = \sum_{i_1,\ldots,i_{2t}} \Delta_\sigma(\vec\imath\,) \ket{i_1,\ldots,i_t}\!\bra{i_{t+1},\ldots,i_{2t}}$, with $\Delta_\sigma(\vec\imath\,) = \delta_{i_{\sigma(1)}i_{\sigma(2)}}\cdots\delta_{i_{\sigma(2t-1)}i_{\sigma(2t)}}$. In a sense, these are a generalization of the usual permutation operators. Eq.~\eqref{eq:Ostates} was derived in Ref.~\cite{harrow2013church} but equivalently follows from Weingarten calculus for the orthogonal group \cite{Collins04,CollinsMat09}.
Consider $t$-th moments of $p_\O(x) = |\vev{x|\O|0}|^2$, where $\ket 0$ is an arbitrary real-valued state. For $\lambda\vdash t$, an integer partition of $t$, the output probabilities are
\begin{equation}
    \Ex_{\O\sim \mu_\O}\bigg[\prod_{i=1}^\ell |p_\O(x_i)|^{\lambda_i}\bigg] = \frac{1}{\prod_{i=0}^{t-1} (d+2i)}\sum_{\sigma \in M_{2t}} \tr\Big(\Pi_\sigma \mathop{\textstyle\bigotimes}\limits_{i=1}^{\ell} \ketbra{x_i}^{\otimes \lambda_i}\Big)\,.
\end{equation}
First consider $\Ex\big[|p_\O(x)|^t\big]$, where we have $\tr(\Pi_\sigma\ketbra{x}^{\otimes t})=1$ for all $\sigma\in M_{2t}$. The result then follows from the fact that there are $\frac{(2t)!}{2^t t!}$ perfect matchings on a set of $2t$ elements. For a general $\lambda\vdash t$ with $\lambda = (\lambda_1,\ldots,\lambda_\ell)$, assuming all states $\ket{x_i}$ are distinct and thus orthogonal, the nonzero $\sigma\in M_{2t}$ are simply those that are perfect matchings on the $\ell$ disjoint subsets of size $2\lambda_i$. Counting the perfect matchings on each subset completes the proof.
\end{proof}

We will primarily just use the two following moments of the output probabilities
\begin{equation}
    \Ex_{\O\sim \mu_\O}\!\!\big[|p_\O(x)|^t\big] = \frac{(2t)!}{2^t t!}\frac{1}{\prod_{i=0}^{t-1} (d+2i)}\,, \quad \Ex_{\O\sim \mu_\O}\!\!\big[|p_\O(x)|^t|p_\O(y)|^t\big] = \left(\frac{(2t)!}{2^t t!}\right)^2\frac{1}{\prod_{i=0}^{2t-1} (d+2i)}\,.
\end{equation}

\begin{definition}[Approximate orthogonal designs] A probability distribution $\nu_\O$ on $\O(d)$ is an $\ep$-approximate orthogonal $t$-design if $\big\| \Phi^{(t)}_{\nu_\O} -  \Phi^{(t)}_{\mu_\O} \|_\diamond \leq \ep$.    
\end{definition}
\ni The literature on orthogonal designs is certainly more sparse than the unitary case, but circuit constructions of approximate orthogonal designs have been given in Refs.~\cite{HHJ20,odonnell2023explicit}

The $t$-th moments of output probabilities of $\O\sim\nu_\O$ for an approximate orthogonal design are upper and lower bounded by the Haar values in \autoref{lem:pOmoms} with additive error $\ep$, which follows from the same proof as in the unitary case in \autoref{lem:approxp}.

\begin{proposition}\label{prop:maxpO}
    Let $\nu_\O$ be an $1/d^{n}$-approximate orthogonal $n$-design, for $\O\sim \nu_\O$ the largest output probability $\max_x p_\O(x)$ obeys
    \begin{equation}
        \Pr_{\O\sim\nu_\O}\left( \max_x p_\O(x) \geq \frac{4n}{d}\right) \leq \frac{2}{d}\,.
    \end{equation}
\end{proposition}
\ni Therefore, with probability exponentially close to 1, the largest output probability is $O(n/d)$.

\begin{proof}[Proof of \autoref{prop:maxpO}]
    Let $\nu_\O$ be an $1/d^t$-approximate orthogonal $t$-design. We again bound a given output probability using higher moments as
    \begin{equation}
        \Pr_{\O\sim \nu_\O} \bigg( p_\O(x) \geq \frac{4n}{d}\bigg) \leq \frac{d^t}{(4n)^t} \Ex_{\O\sim \nu_\O} \big[|p_\O(x)|^t\big] \leq \frac{d^t}{(4n)^t}\frac{2 t^t}{d^t}\leq 2\left(\frac{t}{4n}\right)^t\,,
    \end{equation}
    where we have used that for an $1/d^t$-approximate $t$-design the moments of the output probabilities are bounded as
    \begin{equation}
        \Ex_{\O\sim \nu_\O} \big[|p_\O(x)|^t\big] \leq \frac{\frac{(2t)!}{2^t t!}}{\prod_{i=0}^{t-1}{(d+i)}}+\frac{1}{d^t}\leq \frac{t^t+1}{d^t}\leq 2\frac{t^t}{d^t}\,,
    \end{equation}
    as for all positive integer $t$ we have $\frac{(2t)!}{2^t t!} = \prod_{i=1}^t (2i-1)\leq \big(\frac{1}{t}\sum_{i=1}^t (2i-1)\big)^t= t^t$.
    Therefore, taking a union bound over all $2^n$ output probabilities and taking $t=n$ we find
    \begin{equation}
        \Pr_{\O\sim \nu_\O} \bigg( \max_x p_\O(x) \geq \frac{4n}{d}\bigg) \leq 2^n\, \frac{2}{4^n}= \frac{2}{d}\,.
    \end{equation}
\end{proof}
\ni For Haar random orthogonal matrices, the expected collision probability is
\begin{equation}
    \Ex_{\O\sim \mu_\O}\Big[\Ex_{x\sim p_\O}[p_\O(x)]\Big] = \frac{3}{d+2}.
\end{equation}
By showing that orthogonal designs concentrate around this value, we use a \autoref{prop:maxpO} to show that approximate orthogonal $n$-designs pass the LXEB test.
We choose $b=1.97$ to be consistent with previous theorems, but note that orthogonal designs concentrate around $\approx 3/2^n$
\begin{theorem}\label{thm:odesign}
    An $\ep$-approximate orthogonal $n$-design $\nu_\O$ passes the LXEB test with $b=1.97$, with probability over $\O\sim \nu_\O$ and over $k$ output samples $x_1,\ldots,x_k\sim p_\O$
    \begin{equation}
        \Pr_{\substack{\O\sim \nu_\O\\ x_1,\ldots,x_k\sim p_U}}\Bigg( \frac{1}{k}\sum_{i=1}^k p_\O(x_i)\geq \frac{1.97}{2^n}\bigg) \geq 1 - \exp\left( - \frac{k}{16n}\log\left(\frac{n}{6}\right)\right) - \frac{180}{2^n}\,.
    \end{equation}
\end{theorem}

\begin{proof}[Proof of \autoref{thm:odesign}]
Just as in the unitary design case, we prove that orthogonal designs pass the LXEB test using a bound on the max output probability, a computation of the variance using orthogonal 8-designs, and Bennett's inequality, using the moments of the Haar orthogonal output probabilities in \autoref{lem:pOmoms}.

If $\nu_\O$ is an $1/d^5$-approximate orthogonal 4-design, the collision probability concentrates well
\begin{equation}\label{eq:pOconc}
    \Pr_{\O\sim \nu_\O}\Big( \Big|\Ex_{x\sim p_\O}\big[ p_\O(x)\big] - \frac{3}{d+2}\Big|\geq \frac{\delta}{d}\Big)\leq \frac{25}{\delta^2 d}\,.
\end{equation}
To show this, we first compute the variance over $\mu_\O$ of the collision probability proceeding just as in \autoref{prop:cpconc} using orthogonal output probabilities \autoref{lem:pOmoms} and find
\begin{equation}
    \Var_{\O \sim \mu_\O} \Big( \Ex_x\big[ p_\O(x)\big] \Big) = \frac{24(d-1)}{(d+2)^2(d+4)(d+6)}\,.
\end{equation}
For an $1/d^5$-approximate orthogonal 4-design $\nu_\O$, the centered 2nd moment of the collision probability is bounded as
\begin{equation}
    \Ex_{\O \sim \nu_\O}\bigg[ \Big|\Ex_x\big[ p_\O(x)\big] - \frac{3}{d+2}\Big|^2\bigg] \leq \Var_{\O \sim \mu_\O} \Big( \Ex_x\big[ p_\O(x)\big]\Big) +\frac{1}{d^3}+\frac{6}{d^5}\leq \frac{25}{d^3}\,.
\end{equation}
Eq.~\eqref{eq:pOconc} then follows from Markov's inequality.

We now want to show that the variance over samples concentrates well for approximate orthogonal designs following the same steps as \autoref{lem:varpU}. We first compute the expected variance
\begin{equation}
    \Ex_{\O\sim\mu_\O} \bigg[\Var_{x\sim p_\O}\big(p_\O(x)\big)\bigg] = \frac{6(d-1)}{(d+2)(d+4)(d+6)}\,.
\end{equation}
We then compute the Haar orthogonal variance of the variance $\Var_{\O\sim \mu_\O}\big[ \Var_{x\sim p_\O}(p_\O(x))\big]$ repeatedly using \autoref{lem:pOmoms} and, after a lengthy calculation, we find
\begin{equation}
    \Var_{\O\sim\mu_\O} \bigg[\Var_{x\sim p_\O}\big(p_\O(x)\big)\bigg] = \frac{72 \big(37 d^5+277 d^4-198 d^3-1852 d^2+8360 d-6624\big)}{(d+2)^2(d+4)^2(d+6)^2(d+8)(d+10)(d+12)(d+14)}\,.
\end{equation}
For an $1/d^9$-approximate orthogonal 8-design $\nu_\O$, the sample variance concentrates about the Haar orthogonal value as
\begin{equation}\label{eq:expvarpO}
    \Ex_{\O\sim\nu_\O}\bigg[\Big|\Var_{x\sim p_\O}\big(p_\O(x)\big)-\Ex_{\O\sim\mu_\O}\Big[\Var_{x\sim p_\O}\big(p_\O(x)\big)\Big]\Big|^2\bigg] \leq \Var_{\O\sim\mu_\O} \bigg[\Var_{x\sim p_\O}\big(p_\O(x)\big)\bigg] + \frac{8}{d^5} \leq \frac{2672}{d^5}\,,
\end{equation}
following the same steps as in \autoref{lem:varpU} but for approximate orthogonal designs. Markov's inequality then gives
\begin{equation}\label{eq:prvarpO}
    \Pr_{\O\sim\nu_\O}\bigg(\Var_{x\sim p_\O}\big(p_\O(x)\big)\geq \frac{12}{d^2}\bigg) \leq \Pr_{\O\sim\nu_\O}\!\bigg(\Big|\Var_{x\sim p_\O}\big(p_\O(x)\big)-\!\Ex_{\O\sim\mu_\O}\!\Big[\Var_{x\sim p_\O}\big(p_\O(x)\big)\Big]\Big|^2\geq \frac{36}{d^4}\bigg) \leq \frac{75}{d}\,,
\end{equation}
where we used that Eq.~\eqref{eq:expvarpO} is upper bounded as $\Ex_{\O\sim\mu_\O} \big[\Var_{x\sim p_\O}(p_\O(x))\big]\leq 6/d^2$. 
As in the unitary case in \autoref{thm:polydepthfrml}, we use Bennett's inequality as well as the bound on the max output probability in \autoref{prop:maxpO}, the bound on the variance in Eq.~\eqref{eq:expvarpO}, and the concentration of the collision probability in Eq.~\eqref{eq:pOconc}. Taking a union bound, we find
\begin{equation}
    \Pr_{\substack{\O\sim \nu_\O\\ x_1,\ldots,x_k\sim p_\O}}\Bigg( \frac{1}{k}\sum_{i=1}^k p_\O(x_i)\geq \frac{3}{d+2}-\frac{\delta}{d}-\frac{\delta'}{d}\Bigg) \geq 1 - \exp\bigg( - \frac{k\delta}{8n}\log\bigg(\frac{n\delta}{3}\bigg)\!\bigg) - \frac{75}{d}-\frac{25}{d\delta'^2}-\frac{2}{d}\,,
\end{equation}
Finally, taking $\delta=\delta'=1/2$, the claim then follows.
\end{proof}

\subsection*{Coarse-grained brickwork circuits}
Refs.~\cite{schuster2024random,laracuente2024designs} introduce random quantum circuit ensembles which generate relative-error approximate unitary $t$-designs in depth $O(t\,\mathrm{polylog}(t)\mathrm{log}(n))$.
Similar to the random quantum circuits discussed in the rest of this paper this ensemble draws gates independently in a brickwork architecture.
Here, however, we draw gates from the Haar measure on $SU(4)$ only if they act on qubits inside of blocks of qubits. 
Those gates that act across blocks are chosen to always be the identity.
More precisely, for the ensemble in Ref.~\cite{schuster2024random}, choose $\xi=2\log(nt/\varepsilon)$ and apply independent random quantum circuits of depth $d/2$ to $n/\xi$ blocks arranged in a line.
Then, shift all blocks by $\xi/2$ on this line and apply another round of $n/\xi$ random quantum circuits of depth $d/2$ in the blocks.
The result looks like a single layer of a brickwork circuit where each ``gate'' on $\sim \log(nt/\varepsilon)$ is compiled as a random quantum circuit.
In this setting we have the following theorem:
\begin{theorem}[\cite{schuster2024random,laracuente2024designs}]\label{thm:logdepth}
    Coarse-grained random quantum circuits of depth $O(\mathrm{log}^7(t)t\log(nt/\varepsilon))$ are relative-error $\varepsilon$-approximate unitary designs.
\end{theorem}

As a direct consequence of Theorem~\ref{thm:logdepth} together with Theorem~\ref{thm:lindepthfrml} and Theorem~\ref{thm:polydepthfrml} we obtain the following result:

\begin{corollary}
Coarse-grained random quantum circuits of depth $O(n\,\mathrm{polylog}(n))$ pass the LXEB with probability $\geq 1-O(e^{-k\log(n)/n})-O(1/2^n)$.
Moreover, coarse-grained random quantum circuits of depth $O(\log(n))$ pass the LXEB with probability $\geq 1-O(1/\sqrt{k})-O(1/2^{n})$.
\end{corollary}

Curiously, ensembles of random circuits with orthogonal gates cannot reach approximate orthogonal designs in log depth \cite{schuster2024random,grevink2025will}.
Random orthogonal circuits might still satisfy the bounds as in \autoref{thm:logdepth} but a proof of such a concentration bound will not rely on the design property. 

\bibliographystyle{alpha}
\bibliography{refs}

@article{boixo2018characterizing,
  title={Characterizing quantum supremacy in near-term devices},
  author={Boixo, Sergio and Isakov, Sergei V and Smelyanskiy, Vadim N and Babbush, Ryan and Ding, Nan and Jiang, Zhang and Bremner, Michael J and Martinis, John M and Neven, Hartmut},
  journal={Nat. Phys.},
  volume={14},
  pages={595--600},
  year={2018},
  doi={10.1038/s41567-018-0124-x},
  archivePrefix = {arXiv},
  eprint = {1608.00263},
  primaryClass = {quant-ph}
}

@article{arute2019quantum,
    title = {{Quantum supremacy using a programmable superconducting processor}},
    author = {Arute, Frank and Arya, Kunal and Babbush, Ryan and Bacon, Dave and Bardin, Joseph C. and Barends, Rami and Biswas, Rupak and Boixo, Sergio and others},
    year = {2019},
    journal = {Nature},
    volume = {574},
    pages = {505--510},
    doi = {10.1038/s41586-019-1666-5},
    archivePrefix = {arXiv},
    eprint = {1910.11333},
    primaryClass = {quant-ph}
}

@article{pan2022simulation,
  title = "{Simulation of Quantum Circuits Using the Big-Batch Tensor Network Method}",
  author = {Pan, Feng and Zhang, Pan},
  journal = {Phys. Rev. Lett.},
  volume = {128},
  issue = {3},
  pages = {030501},
  numpages = {6},
  year = {2022},
  publisher = {American Physical Society},
  doi = {10.1103/PhysRevLett.128.030501}
}

@article{pan2022solving,
  title = "{Solving the Sampling Problem of the Sycamore Quantum Circuits}",
  author = {Pan, Feng and Chen, Keyang and Zhang, Pan},
  journal = {Phys. Rev. Lett.},
  volume = {129},
  issue = {9},
  pages = {090502},
  numpages = {6},
  year = {2022},
  publisher = {American Physical Society},
  doi = {10.1103/PhysRevLett.129.090502}
}

@article{morvan2023phase,
      title={Phase transition in Random Circuit Sampling}, 
      author={A. Morvan and B. Villalonga and X. Mi and S. Mandrà and A. Bengtsson and P. V. Klimov and Z. Chen and S. Hong and C. Erickson and I. K. Drozdov and J. Chau et. al.},
      year={2023},
      eprint={2304.11119},
      archivePrefix={arXiv},
      primaryClass={quant-ph},
      journal={arXiv preprint arXiv:2304.11119}
}

@article{zhu2021quantum,
    author={Qingling Zhu and Sirui Cao and Fusheng Chen and Ming-Cheng Chen and Xiawei Chen and Tung-Hsun Chung and Hui Deng and Yajie Du and others},
    title = "{Quantum computational advantage via 60-qubit 24-cycle random circuit sampling}",
    eprint = "2109.03494",
    archivePrefix = "arXiv",
    primaryClass = "quant-ph",
    doi = "10.1016/j.scib.2021.10.017",
    journal = "Sci. Bull.",
    volume = "67",
    pages = "240--245",
    year = "2022"
}

@article{wu2021strong,
    author={Wu, Yulin and Bao, Wan-Su and Cao, Sirui and Chen, Fusheng and Chen, Ming-Cheng and Chen, Xiawei and Chung, Tung-Hsun and Deng, Hui and others},
    title = "{Strong quantum computational advantage using a superconducting quantum processor}",
    eprint = "2106.14734",
    archivePrefix = "arXiv",
    primaryClass = "quant-ph",
    doi = "10.1103/PhysRevLett.127.180501",
    journal = "Phys. Rev. Lett.",
    volume = "127",
    year = "2021"
}

@article{decross2025comp,
  title = "{Computational Power of Random Quantum Circuits in Arbitrary Geometries}",
  author = {DeCross, M. and Haghshenas, R. and Liu, M. and Rinaldi, E. and Gray, J. and Alexeev, Y. and Baldwin, C. H. and Bartolotta, J. P. and others},
  journal = {Phys. Rev. X},
  volume = {15},
  issue = {2},
  pages = {021052},
  numpages = {39},
  year = {2025},
  publisher = {American Physical Society},
  doi = {10.1103/PhysRevX.15.021052},
}

@article{zhong2020quantum,
  title={Quantum computational advantage using photons},
  author={Zhong, Han-Sen and Wang, Hui and Deng, Yu-Hao and Chen, Ming-Cheng and Peng, Li-Chao and Luo, Yi-Han and Qin, Jian and Wu, Dian and Ding, Xing and others},
  journal={Science},
  volume={370},
  number={6523},
  pages={1460--1463},
  year={2020}
}

@article{madsen2022quantum,
  title={Quantum computational advantage with a programmable photonic processor},
  author={Madsen, Lars S and Laudenbach, Fabian and Askarani, Mohsen Falamarzi and Rortais, Fabien and Vincent, Trevor and Bulmer, Jacob FF and Miatto, Filippo M and Neuhaus, Leonhard and others},
  journal={Nature},
  volume={606},
  number={7912},
  pages={75--81},
  year={2022}
}

@article{gao2025establishing,
  title = "{Establishing a New Benchmark in Quantum Computational Advantage with 105-qubit Zuchongzhi 3.0 Processor}",
  author = {Gao, Dongxin and Fan, Daojin and Zha, Chen and Bei, Jiahao and Cai, Guoqing and Cai, Jianbin and Cao, Sirui and Chen, Fusheng and others},
  journal = {Phys. Rev. Lett.},
  volume = {134},
  issue = {9},
  pages = {090601},
  numpages = {7},
  year = {2025},
  publisher = {American Physical Society},
  doi = {10.1103/PhysRevLett.134.090601},
}

@article{liu2025certified,
    author={Liu, Minzhao and Shaydulin, Ruslan and Niroula, Pradeep and DeCross, Matthew and Hung, Shih-Han and Kon, Wen Yu and Cervero-Mart{\'\i}n, Enrique and Chakraborty, Kaushik and others},
    title = "{Certified randomness using a trapped-ion quantum processor}",
    eprint = "2503.20498",
    archivePrefix = "arXiv",
    primaryClass = "quant-ph",
    doi = "10.1038/s41586-025-08737-1",
    journal = "Nature",
    volume = "640",
    number = "8058",
    pages = "343--348",
    year = "2025"
}

@article{liu2025certified2,
  title={Certified randomness amplification by dynamically probing remote random quantum states},
  author={Liu, Minzhao and Niroula, Pradeep and DeCross, Matthew and Foreman, Cameron and Kon, Wen Yu and Primaatmaja, Ignatius William and Allman, M.~S. and Campora III, J.~P. and Isanaka, Akhil and others},
  journal={arXiv preprint arXiv:2511.03686},
  year={2025}
}

@article{BHH12,
	author = {{Brand{\~a}o}, F.~G.~S.~L. and {Harrow}, A.~W. and {Horodecki}, M.},
	title = "{Local Random Quantum Circuits are Approximate Polynomial-Designs}",
	journal = {Commun. Math. Phys.},
	archivePrefix = "arXiv",
	eprint = {1208.0692},
	primaryClass = "quant-ph",
	year = 2016,
	volume = 346,
	pages = {397},
	doi = {10.1007/s00220-016-2706-8}
}

@article{BF13,
   author = {{Brown}, W. and {Fawzi}, O.},
    title = "{Decoupling with random quantum circuits}",
    journal="Comm. Math. Phys.",
	year="2015",
	volume="340",
	pages="867",
	doi="10.1007/s00220-015-2470-1",
	archivePrefix = "arXiv",
  	 eprint = {1307.0632},
 	primaryClass = "quant-ph",
}

@article{RQCstatmech,
      author         = "Hunter-Jones, Nicholas",
      title          = "{Unitary designs from statistical mechanics in random quantum circuits}",
      year           = "2019",
      eprint         = "1905.12053",
      archivePrefix  = "arXiv",
      primaryClass   = "quant-ph",
      journal        = "arXiv preprint arXiv:1905.12053"
}

@article{HHJ20,
    author = "Haferkamp, Jonas and Hunter-Jones, Nicholas",
    title = "{Improved spectral gaps for random quantum circuits: Large local dimensions and all-to-all interactions}",
    eprint = "2012.05259",
    archivePrefix = "arXiv",
    primaryClass = "quant-ph",
    doi = "10.1103/PhysRevA.104.022417",
    journal = "Phys. Rev. A",
    volume = "104",
    number = "2",
    pages = "022417",
    year = "2021"
}

@article{haferkamp2022random,
    author = "Haferkamp, Jonas",
    title = "{Random quantum circuits are approximate unitary $t$-designs in depth $O\left(nt^{5+o(1)}\right)$}",
    eprint = "2203.16571",
    archivePrefix = "arXiv",
    primaryClass = "quant-ph",
    doi = "10.22331/q-2022-09-08-795",
    journal = "Quantum",
    volume = "6",
    pages = "795",
    year = "2022"
}

@inproceedings{chen2025incompressibility,
  title={Incompressibility and spectral gaps of random circuits},
  author={Chen, Chi-Fang and Haah, Jeongwan and Haferkamp, Jonas and Liu, Yunchao and Metger, Tony and Tan, Xinyu},
  booktitle={2025 IEEE 66th Annual Symposium on Foundations of Computer Science (FOCS)},
  pages={1304--1312},
  year={2025},
  organization={IEEE}
}

@article{schuster2024random,
    author = "Schuster, Thomas and Haferkamp, Jonas and Huang, Hsin-Yuan",
    title = "{Random unitaries in extremely low depth}",
    eprint = "2407.07754",
    archivePrefix = "arXiv",
    primaryClass = "quant-ph",
    doi = "10.1126/science.adv8590",
    journal = "Science",
    volume = "389",
    number = "6755",
    pages = "adv8590",
    year = "2025"
}

@article{BH13,
 author = {Brand\~{a}o, Fernando G. S. L. and Horodecki, Michal},
 title = {Exponential Quantum Speed-ups Are Generic},
 journal = {Quantum Info. Comput.},
 volume = {13},
 year = {2013},
 pages = {901},
 archivePrefix = {arXiv},
 eprint = {1010.3654},
 primaryClass = {quant-ph}
}

@article{nietner2023average,
    author = "Nietner, Alexander and Ioannou, Marios and Sweke, Ryan and Kueng, Richard and Eisert, Jens and Hinsche, Marcel and Haferkamp, Jonas",
    title = "{On the average-case complexity of learning output distributions of quantum circuits}",
    eprint = "2305.05765",
    archivePrefix = "arXiv",
    primaryClass = "quant-ph",
    doi = "10.22331/q-2025-10-13-1883",
    journal = "Quantum",
    volume = "9",
    pages = "1883",
    year = "2025"
}

@article{webb2015clifford,
    author = {Zak Webb},
    title = "{The Clifford group forms a unitary 3-design}",
    journal = {Quantum Info. Comput.},
    volume = {16},
    pages = {1379},
    year = {2016},
    archivePrefix = "arXiv",
    eprint = {1510.02769},
    primaryClass = "quant-ph"
}

@article{zhu2017multiqubit,
  title = {Multiqubit Clifford groups are unitary 3-designs},
  author = {Zhu, Huangjun},
  journal = {Phys. Rev. A},
  volume = {96},
  issue = {6},
  pages = {062336},
  numpages = {7},
  year = {2017},
  publisher = {American Physical Society},
  doi = {10.1103/PhysRevA.96.062336},
  archivePrefix = "arXiv",
  eprint = {1510.02619},
  primaryClass = "quant-ph"
}

@article{NRVH16,
      author         = "Nahum, Adam and Ruhman, Jonathan and Vijay, Sagar and
                        Haah, Jeongwan",
      title          = "{Quantum Entanglement Growth Under Random Unitary
                        Dynamics}",
      journal        = "Phys. Rev.",
      volume         = "X7",
      year           = "2017",
      pages          = "031016",
      doi            = "10.1103/PhysRevX.7.031016",
      eprint         = "1608.06950",
      archivePrefix  = "arXiv",
      primaryClass   = "cond-mat.stat-mech",
      SLACcitation   = "%%CITATION = ARXIV:1608.06950;%%"
}

@article{NVH17,
      author         = "Nahum, Adam and Vijay, Sagar and Haah, Jeongwan",
      title          = "{Operator Spreading in Random Unitary Circuits}",
      journal        = "Phys. Rev.",
      volume         = "X8",
      year           = "2018",
      pages          = "021014",
      doi            = "10.1103/PhysRevX.8.021014",
      eprint         = "1705.08975",
      archivePrefix  = "arXiv",
      primaryClass   = "cond-mat.str-el",
      SLACcitation   = "%%CITATION = ARXIV:1705.08975;%%"
}

@article{vonKey17,
      author         = "von Keyserlingk, Curt and Rakovszky, Tibor and Pollmann,
                        Frank and Sondhi, Shivaji",
      title          = "{Operator hydrodynamics, OTOCs, and entanglement growth
                        in systems without conservation laws}",
      journal        = "Phys. Rev.",
      volume         = "X8",
      year           = "2018",
      pages          = "021013",
      doi            = "10.1103/PhysRevX.8.021013",
      eprint         = "1705.08910",
      archivePrefix  = "arXiv",
      primaryClass   = "cond-mat.str-el",
      SLACcitation   = "%%CITATION = ARXIV:1705.08910;%%"
}

@article{Collins04,
   	author = {{Collins}, B. and {{\'S}niady}, P.},
    title = "{Integration with Respect to the Haar Measure on Unitary, Orthogonal and Symplectic Group}",
  	journal = {Commun. Math. Phys.},
  	archivePrefix = "arXiv",
   	eprint = {math-ph/0402073},
    year = 2006,
   	volume = 264,
    pages = {773},
    doi = {10.1007/s00220-006-1554-3},
}

@article{CollinsMat09,
	author = {Benoît Collins and Sho Matsumoto},
	title = "{On some properties of orthogonal Weingarten functions}",
	journal = {J. Math. Phys.},
	volume = {50},
	pages = {113516},
	year = {2009},
	archivePrefix = "arXiv",
	eprint = {0903.5143},
	primaryClass = "math-ph",
	doi = {10.1063/1.3251304},
}

@article{Bennett62,
 author = {George Bennett},
 journal = {J. Am. Stat. Assoc.},
 number = {297},
 pages = {33--45},
 title = {Probability Inequalities for the Sum of Independent Random Variables},
 volume = {57},
 year = {1962}
}

@book{concineq,
  title={Concentration Inequalities: A Nonasymptotic Theory of Independence},
  author={Boucheron, S. and Lugosi, G. and Massart, P.},
  isbn={9780199535255},
  lccn={2012277339},
  year={2013},
  publisher={OUP Oxford}
}

@article{hangleiter2023comp,
    author = "Hangleiter, Dominik and Eisert, Jens",
    title = "{Computational advantage of quantum random sampling}",
    eprint = "2206.04079",
    archivePrefix = "arXiv",
    primaryClass = "quant-ph",
    doi = "10.1103/RevModPhys.95.035001",
    journal = "Rev. Mod. Phys.",
    volume = "95",
    number = "3",
    pages = "035001",
    year = "2023"
}

@article{kliesch2021theory,
    author = "Kliesch, Martin and Roth, Ingo",
    title = "{Theory of Quantum System Certification}",
    eprint = "2010.05925",
    archivePrefix = "arXiv",
    primaryClass = "quant-ph",
    doi = "10.1103/PRXQuantum.2.010201",
    journal = "PRX Quantum",
    volume = "2",
    number = "1",
    pages = "010201",
    year = "2021"
}

@article{aaronson2020classical,
    author = "Aaronson, Scott and Gunn, Sam",
    title = "{On the Classical Hardness of Spoofing Linear Cross-Entropy Benchmarking}",
    eprint = "1910.12085",
    archivePrefix = "arXiv",
    primaryClass = "quant-ph",
    doi = "10.4086/toc.2020.v016a011",
    journal = "Theor. Comp.",
    volume = "16",
    number = "1",
    pages = "1--8",
    year = "2020"
}

@article{gao2021limitations,
    author = "Gao, Xun and Kalinowski, Marcin and Chou, Chi-Ning and Lukin, Mikhail D. and Barak, Boaz and Choi, Soonwon",
    title = "{Limitations of Linear Cross-Entropy as a Measure for Quantum Advantage}",
    eprint = "2112.01657",
    archivePrefix = "arXiv",
    primaryClass = "quant-ph",
    doi = "10.1103/PRXQuantum.5.010334",
    journal = "PRX Quantum",
    volume = "5",
    number = "1",
    pages = "010334",
    year = "2024",
    eprint={2112.01657},
    archivePrefix={arXiv},
    primaryClass={quant-ph}
}

@inproceedings{barak2020spoofing,
  author = {Barak, Boaz and Chou, Chi-Ning and Gao, Xun},
  title = {{Spoofing Linear Cross-Entropy Benchmarking in Shallow Quantum Circuits}},
  booktitle =	{12th Innovations in Theoretical Computer Science Conference (ITCS 2021)},
  pages =	{30:1--30:20},
  series = {LIPIcs},
  year = {2021},
  volume = {185},
  publisher =	{Schloss Dagstuhl},
  doi = {10.4230/LIPIcs.ITCS.2021.30}
}

@inproceedings{aharonov2023poly,
    author = {Aharonov, Dorit and Gao, Xun and Landau, Zeph and Liu, Yunchao and Vazirani, Umesh},
    title = {A Polynomial-Time Classical Algorithm for Noisy Random Circuit Sampling},
    year = {2023},
    publisher = {Association for Computing Machinery},
    doi = {10.1145/3564246.3585234},
    booktitle = {Proceedings of the 55th Annual ACM Symposium on Theory of Computing},
    pages = {945--957},
    numpages = {13},
    series = {STOC 2023}
}

@article{bassirian2021certified,
    author = "Bassirian, Roozbeh and Bouland, Adam and Fefferman, Bill and Gunn, Sam and Tal, Avishay",
    title = "{On Certified Randomness from Fourier Sampling or Random Circuit Sampling}",
    eprint = "2111.14846",
    archivePrefix = "arXiv",
    primaryClass = "quant-ph",
    doi = "10.22331/q-2026-02-10-2002",
    journal = "Quantum",
    volume = "10",
    pages = "2002",
    year = "2026"
}

@article{dalzell2022anticon,
    author = "Dalzell, Alexander M. and Hunter-Jones, Nicholas and Brand{\~a}o, Fernando G. S. L.",
    title = "{Random Quantum Circuits Anticoncentrate in Log Depth}",
    eprint = "2011.12277",
    archivePrefix = "arXiv",
    primaryClass = "quant-ph",
    doi = "10.1103/PRXQuantum.3.010333",
    journal = "PRX Quantum",
    volume = "3",
    number = "1",
    pages = "010333",
    year = "2022"
}

@article{dalzell2024noise,
    author = "Dalzell, Alexander M. and Hunter-Jones, Nicholas and Brand{\~a}o, Fernando G. S. L.",
    title = "{Random Quantum Circuits Transform Local Noise into Global White Noise}",
    eprint = "2111.14907",
    archivePrefix = "arXiv",
    primaryClass = "quant-ph",
    doi = "10.1007/s00220-024-04958-z",
    journal = "Commun. Math. Phys.",
    volume = "405",
    number = "3",
    pages = "78",
    year = "2024"
}

@article{ware2023sharp,
  title={A sharp phase transition in linear cross-entropy benchmarking},
  author={Ware, Brayden and Deshpande, Abhinav and Hangleiter, Dominik and Niroula, Pradeep and Fefferman, Bill and Gorshkov, Alexey V and Gullans, Michael J},
  journal={arXiv preprint arXiv:2305.04954},
  year={2023}
}

@article{heinrich2023randomized,
    author = {{Heinrich}, Markus and {Kliesch}, Martin and {Roth}, Ingo},
    title = "{Randomized benchmarking with random quantum circuits}",
    journal = {arXiv e-prints arXiv:2212.06181},
    year = 2022,
    eid = {arXiv:2212.06181},
    doi = {10.48550/arXiv.2212.06181},  
    archivePrefix = {arXiv},
    eprint = {2212.06181},
    primaryClass = {quant-ph}
}

@inproceedings{laracuente2024designs,
  author = {LaRacuente, Nicholas and Leditzky, Felix},
  title =	{{Approximate Unitary $k$-Designs from Shallow, Low-Communication Circuits}},
  booktitle =	{16th Innovations in Theoretical Computer Science Conference (ITCS 2025)},
  pages =	{69:1--69:2},
  series = {Leibniz International Proceedings in Informatics (LIPIcs)},
  year = {2025},
  volume = {325},
  eprint = "2407.07876",
  archivePrefix = "arXiv",
  primaryClass = "quant-ph"
}

@article{MHJ23,
    author = "Mittal, Shivan and Hunter-Jones, Nicholas",
    title = "{Local random quantum circuits form approximate designs on arbitrary architectures}",
    eprint = "2310.19355",
    archivePrefix = "arXiv",
    primaryClass = "quant-ph",
    year = "2023",
    journal = {arXiv e-prints arXiv:2310.19355},
    eid = {arXiv:2310.19355},
}

@article{harrow2013church,
    author = "Harrow, Aram W.",
    title = "{The Church of the Symmetric Subspace}",
    eprint = "1308.6595",
    archivePrefix = "arXiv",
    primaryClass = "quant-ph",
    year = "2013",
    journal = {arXiv e-prints arXiv:1308.6595},
}

@article{odonnell2023explicit,
    author = {{O'Donnell}, Ryan and {Servedio}, Rocco A. and {Paredes}, Pedro},
    title = "{Explicit orthogonal and unitary designs}",
    journal = {arXiv e-prints arXiv:2310.13597},
    year = 2023,
    eid = {arXiv:2310.13597},
    doi = {10.48550/arXiv.2310.13597},
    archivePrefix = {arXiv},
    eprint = {2310.13597},
    primaryClass = {cs.CC}
}

@article{heinrich2025anti,
  title={Anti-concentration is (almost) all you need},
  author={Heinrich, Markus and Haferkamp, Jonas and Roth, Ingo and Helsen, Jonas},
  journal={arXiv preprint arXiv:2510.23719},
  year={2025}
}

@article{laracuente2025quantum,
  title="{Quantum Relative Entropy Decay Composition Yields Shallow, Unstructured k-Designs}",
  author={Laracuente, Nicholas},
  journal={arXiv preprint arXiv:2510.08537},
  year={2025}
}

@article{grevink2025will,
  title="{Will it glue? On short-depth designs beyond the unitary group}",
  author={Grevink, Lorenzo and Haferkamp, Jonas and Heinrich, Markus and Helsen, Jonas and Hinsche, Marcel and Schuster, Thomas and Zimbor{\'a}s, Zolt{\'a}n},
  journal={arXiv preprint arXiv:2506.23925},
  year={2025}
}

@article{manole2025howmuch,
    author = "Manole, Tudor and Mark, Daniel K. and Gong, Wenjie and Ye, Bingtian and Polyanskiy, Yury and Choi, Soonwon",
    title = "{How much can we learn from quantum random circuit sampling?}",
    eprint = "2510.09919",
    archivePrefix = "arXiv",
    primaryClass = "quant-ph",
    journal={arXiv preprint arXiv:2510.09919},
    year = "2025"
}

@inproceedings{aaronson2023certified,
author = {Aaronson, Scott and Hung, Shih-Han},
title = "{Certified Randomness from Quantum Supremacy}",
year = {2023},
isbn = {9781450399135},
doi = {10.1145/3564246.3585145},
booktitle = {Proceedings of the 55th Annual ACM Symposium on Theory of Computing},
pages = {933--944},
numpages = {12},
series = {STOC 2023}
}

@article{bremner2016average,
    author = "Bremner, Michael J. and Montanaro, Ashley and Shepherd, Dan J.",
    title = "{Average-Case Complexity Versus Approximate Simulation of Commuting Quantum Computations}",
    eprint = "1504.07999",
    archivePrefix = "arXiv",
    primaryClass = "quant-ph",
    doi = "10.1103/PhysRevLett.117.080501",
    journal = "Phys. Rev. Lett.",
    volume = "117",
    number = "8",
    pages = "080501",
    year = "2016"
}

@article{morimae2017hardness,
    author = "Morimae, Tomoyuki",
    title = "{Hardness of classically sampling one clean qubit model with constant total variation distance error}",
    eprint = "1704.03640",
    archivePrefix = "arXiv",
    primaryClass = "quant-ph",
    doi = "10.1103/PhysRevA.96.040302",
    journal = "Phys. Rev. A",
    volume = "96",
    pages = "040302",
    year = "2017"
}

@inproceedings{aaronson2010computational,
    author = {Aaronson, Scott and Arkhipov, Alex},
    title = {The computational complexity of linear optics},
    year = {2011},
    isbn = {9781450306911},
    publisher = {Association for Computing Machinery},
    doi = {10.1145/1993636.1993682},
    booktitle = {Proceedings of the Forty-Third Annual ACM Symposium on Theory of Computing},
    pages = {333--342},
    numpages = {10},
    series = {STOC 2011}
}

\end{document}